\documentclass[a4paper, 12pt]{article}

\usepackage[top = 2.5cm, bottom = 2.5cm, left = 2.4cm, right = 2.4cm]{geometry}
\usepackage{amsfonts, amsmath, amsthm, amssymb, mathtools, cases, bbm, array}
\usepackage{lmodern, indentfirst, setspace, changepage, fancyhdr}
\usepackage{tikz, pgfplots, float, subcaption}
\usepackage[title]{appendix}
\usepackage[english]{babel}
\usepackage{cite}

\usepackage{hyperref}
\hypersetup{
	colorlinks = true,
	linkcolor = blue,
	citecolor = blue,
}

\newcommand{\argmin}[1]{\underset{#1}{\text{argmin}}\hspace{1.5pt}}
\newcommand{\map}[3]{#1 : #2 \longrightarrow #3}

\newcommand{\sset}[2]{\{#1 : #2\}}

\newcommand{\double}{\ \ }
\newcommand{\bL}{\bar{L}}
\newcommand{\bX}{\bar{X}}
\newcommand{\bx}{\bar{x}}

\newtheorem{theorem}{Theorem}

\newtheorem{proposition}[theorem]{Proposition}

\newtheorem{remark}{Remark}

\newcommand{\vpo}{\vspace{1ex}}

\newcommand{\Sii}{{\Longleftrightarrow}}

\newcommand{\nonu}{\nonumber}
\newcommand{\edoc}{\end{document}}
\newcommand{\beq}{\begin{equation}}  
\newcommand{\eeq}{\end{equation}}    
\newcommand{\bea}{\begin{align}}
\newcommand{\eea}{\end{align}}
\newcommand{\beqn}{\begin{eqnarray}}
\newcommand{\eeqn}{\end{eqnarray}}
\newcommand{\beqns}{\begin{eqnarray*}}
\newcommand{\eeqns}{\end{eqnarray*}}
\newcommand{\bct}{\begin{center}}
\newcommand{\ect}{\end{center}}
\newcommand{\btmz}{\begin{itemize}}
\newcommand{\etmz}{\end{itemize}}
\newcommand{\benum}{\begin{enumerate}}
\newcommand{\eenum}{\end{enumerate}}

\newcommand{\defequal}{:=}   

\newcommand{\image}{\mbox{\rm{Im\kern.043333em}}}

\newcommand{\norm}[1]{\| #1 \|}         

\newcommand{\R}{{\mathbb R}}

\usetikzlibrary{automata, arrows, positioning, calc, external, babel, backgrounds, matrix, shapes}
\usepgfplotslibrary{fillbetween}
\pgfplotsset{
	compat = 1.16,
	ticklabel style = {font = \footnotesize},
	every axis/.append style = {
		grid style = {dashed, gray, opacity = 0.2},
		label style = {font = \footnotesize}, 
		width = \columnwidth,
		height = 0.618 * 1 * \columnwidth
	}
}

\definecolor{britishracinggreen}{rgb}{0.0, 0.26, 0.15}
\definecolor{bostonuniversityred}{rgb}{0.8, 0.0, 0.0}
\definecolor{ceruleanblue}{rgb}{0.16, 0.32, 0.75}
\definecolor{airforceblue}{rgb}{0.36, 0.54, 0.66}
\definecolor{cadmiumgreen}{rgb}{0.0, 0.42, 0.24}
\definecolor{ao(english)}{rgb}{0.0, 0.5, 0.0}
\definecolor{coolblack}{rgb}{0.0, 0.18, 0.39}
\definecolor{byzantine}{rgb}{0.74, 0.2, 0.64}
\definecolor{alizarin}{rgb}{0.82, 0.1, 0.26}
\definecolor{arsenic}{rgb}{0.23, 0.27, 0.29}
\definecolor{cobalt}{rgb}{0.0, 0.28, 0.67}
\definecolor{amber}{rgb}{1.0, 0.75, 0.0}

\title{Dynamic load balancing for cloud systems\\under heterogeneous setup delays\vspace{\baselineskip}}

\author{
	\begin{tabular}{ccc}
		\normalsize Fernando Paganini & \hspace{1cm} & \normalsize Diego Goldsztajn \\ 
		\small Universidad ORT Uruguay & \hspace{1cm} & \small Universidad ORT Uruguay \\
		\scriptsize\texttt{paganini@ort.edu.uy} & \hspace{1cm} & \scriptsize\texttt{goldsztajn@ort.edu.uy} \\
	\end{tabular}
}

\date{\vspace{\baselineskip} August 12, 2025}

\begin{document}

	
\maketitle

\noindent\rule{\textwidth}{1pt}

\vspace{2\baselineskip}

\onehalfspacing

\begin{adjustwidth}{0.8cm}{0.8cm}
	\begin{center}
		\textbf{Abstract}
	\end{center}
	
	\vspace{0.3\baselineskip}
	
	\noindent	
	We consider a distributed cloud service deployed at a set of distinct server pools. Arriving jobs are classified into heterogeneous types, in accordance with their setup times which are differentiated at each of the pools. A dispatcher for each job type controls the balance of load between pools, based on decentralized feedback. The system of rates and queues is modeled by a fluid differential equation system, and analyzed via convex optimization. A first, myopic policy is proposed, based on task delay-to-service. 
	Under a simplified dynamic fluid queue model, we prove global convergence to an equilibrium point which minimizes the mean setup time; however queueing delays are incurred with this method. A second proposal is then developed based on proximal optimization, which explicitly models the setup queue and is proved to reach an optimal equilibrium, devoid of queueing delay. Results are demonstrated through a simulation example.
	
	\vspace{\baselineskip}
	
	\small{\noindent \textit{Key words:} cloud networks, load balancing, heterogeneous delays, fluid-based control.}
	
	\vspace{0.3\baselineskip}
	
	\small{\noindent The present paper was supported by AFOSR-US under grant \#FA9550-23-1-0350 and ANII-Uruguay under fellowship PD\_NAC\_2024\_182118.} 
\end{adjustwidth}

\newpage


\section{Introduction}\label{sec.intro}

The evolution of networking and computation involves an ever-increasing parallelization. Jobs are transported through the global Internet, to be processed among millions of server units located at multiple distributed locations. As these cloud infrastructures develop, so must the algorithms required to orchestrate such scattered resources in an efficient manner. 

A core technology in this domain is \emph{load balancing}, with a long history (e.g., \cite{tantawi1985optimal}). Job requests arriving at different network locations must be \emph{dispatched} to an appropriate server, with the general objective of reducing service latency. A classical policy of general use is to route to the server with the shortest queue \cite{ephremides1980simple}; more recently, cloud computing has motivated alternate policies with a smaller burden of information exchange when scaling to a large server population (see \cite{der2022scalable} and references therein).  

A complicating factor is \emph{heterogeneity} of server resources: they could vary in  their network distance to the demand source, their access to {data} relevant for the computation, or their processing power. We will ignore the latter point and assume servers with standardized capacity, but capture the other sources of heterogeneity, manifested in \emph{setup times} needed before an actual computation may take place. Servers are deployed in \emph{pools} indexed by $j$, and tasks classified in \emph{types} $i$ based on their location or data characteristics; the setup delay will be a function of $(i,j)$. 

The control problem is as follows: a dispatcher for type $i$ must distribute the incoming rate of tasks among the pools, taking into account their congestion state and attempting to minimize the setup latency involved. Fundamentally, dispatchers are not allowed to communicate, they must make decentralized decisions based on feedback from pool queues and setup information for their own tasks. 

We employ fluid, differential equation models for task rates and queues, and rely extensively on convex optimization to design decentralized control rules for dispatchers, with provable convergence to an equilibrium which minimizes a latency cost. This methodology has proven successful to attack a variety of problems in control of networks: congestion control, routing, medium access, their cross-layer interactions, etc., see e.g. \cite{low2002internet,srikant2013communication}. 

In \cite{paganini2019optimization} we brought these fluid and optimization tools 
to study the combination of shortest-queue load balancing with certain policies for task scheduling and speed scaling at computer clusters. Our current model has the following differences with this earlier work:
\btmz
\item In \cite{paganini2019optimization}, following \cite{wang2016maptask}, heterogeneity was captured by differences in the service \emph{rate} for different data types. This is not very realistic, cloud systems tend to have standardized capacity servers. An additive setup delay more accurately represents the situation of e.g. data requirements which must be fulfilled prior to service. 
\item The service model at clusters in \cite{paganini2019optimization} implicitly assumed a \emph{processor sharing} discipline, where the entire pool capacity is applied to the current tasks, irrespective of their number. A more practical model is an individual assignment of tasks to servers, which means a pool loaded below its capacity will be underutilized. 
\etmz
Additive latencies make the current problem similar to load balancing in the transportation literature \cite{sheffi1985urban};  indeed, the first part of our results has a counterpart, with mathematical differences, with those obtained recently in \cite{paganini2025tcns} for selfish routing in  electrical vehicle charging infrastructures. 

The paper is organized as follows: in Section \ref{sec.formu} we present our initial fluid models for load balancing. In Section \ref{sec.myopic} a first load balancing strategy is proposed, where tasks are routed myopically to pools with the smallest delay-to-service; the dynamics are analyzed through a convex optimization problem and its Lagrange dual. Some limitations in modeling, and in performance, are identified at the end of the section. Sections \ref{sec.setup} and \ref{sec.prox} develop a richer model and a new control proposal, based on proximal optimization, that more accurately 
represents the problem and achieves the desired performance. A brief experimental demonstration is presented in Section \ref{sec.simulations}, and conclusions are given in Section \ref{sec.concl}. Some proofs are deferred to the Appendix.

\section{Fluid model for load balancing}\label{sec.formu}

We consider a set of server pools, indexed by $j=1,\ldots, n$. Each is identified with a service location, and has a capacity of $c_j$ servers. For simplicity, we assume each server has unit service rate in tasks/sec. 

Tasks arriving into the system are classified in \emph{types} $i=1, \ldots, m$, which may distinguish the location from where the request arises, and/or the data required to perform the task. This distinction is reflected in heterogeneous times for the \emph{setup} of a task before it can be served at the corresponding pool. In our initial model, we denote by $\tau_{ij}$ the mean setup time for a request of type $i$ in pool $j$.

We denote by $r_i$ the arrival rate of tasks of type $i$, an exogenous quantity. We assume there is a \emph{dispatcher} for each type, which makes \emph{load balancing} decisions based on information on congestion at the different pools, and the respective setup times. These decisions are captured by variables $x_{ij}(t)$, representing the rate of tasks of type $i$ that are routed to pool $j$; dependence on continuous time $t$ will be left implicit henceforth. The routing constraints to be satisfied at all times are:
\begin{align}\label{eq.xij}
	x_{ij}\geq 0, \quad \sum_{j=1}^n x_{ij} = r_i, \quad i = 1,\ldots, m.
\end{align}
We denote by $X \in \R_+^{m\times n}$ the matrix with entries $x_{ij}$. 

On the server side, we denote by $q_j(t)$ the number of tasks currently assigned to pool $j$, whether in service or awaiting service. Since the pool capacity is $c_j$, the number of active servers will be $\min(q_j,c_j)$. This coincides with the  task departure rate at pool $j$, since servers have unit service rate. Balancing the total arrival and departure rates we reach our first dynamic model for the task queues, treated as fluid (continuous) variables:
\begin{align}\label{eq.queue}
	\dot{q}_{j} =  \sum_{i=1}^m x_{ij} - \min(q_j,c_j), \quad j = 1,\ldots, n.
\end{align}

\begin{remark}\label{rem.delay}
	In eq. \eqref{eq.queue}, the input rates are applied immediately to the service queues, when in fact they should arrive after a setup time. If these times are assumed to be fixed (deterministic), a more precise description would be the delay-differential equation model
	\begin{align}\label{eq.delaydiff}
		\dot{q}_{j} =  \sum_{i=1}^m x_{ij}(t-\tau_{ij}) - \min(q_j,c_j), \quad j = 1,\ldots, n.
	\end{align}
	The results below on equilibrium apply indistinctly to both models. For dynamic convergence studies, we will use the simplified model \eqref{eq.queue}. An alternative, where tasks under setup are modeled as a separate fluid queue, will be considered in Sections \ref{sec.setup} and \ref{sec.prox}.
\end{remark}

\section{Myopic balancing based on delay to service} \label{sec.myopic}

In this section we assume that dispatch decisions are based on the \emph{delay-to-service}; this is comprised of the setup time $\tau_{ij}$ and potentially a waiting time at pool $j$, which we denote by $\mu_j$. Waiting occurs when queue occupation exceeds pool capacity; the excess tasks $[q_j-c_j]^+$ must be cleared at the overall service rate, in this case $c_j$. We therefore have the delay model 
\begin{align}\label{eq.mu}
	\mu_j(q_j) = 
	\frac{[q_j-c_j]^+}{c_j} = \left[\frac{q_j}{c_j}-1\right]^+.
\end{align}
The myopic rule for load balancing is to dispatch each task of type $i$ to the pool with shortest delay-to-service:
\begin{align}
	\label{eq.myopic}
	j^* \in \arg\min_j\{\tau_{ij}+\mu_j\},
\end{align}
with a tie-breaking rule for the case of equal delays. 
Mathematically, this results in a discontinuous dynamics, as rates $x_{ij}$ switch from one pool to another. To avoid this, preserving the essence, we will work here with a ``soft-min" version of the myopic load balancing rule. For this purpose, consider 
\begin{align}\label{eq.logsumexp}
	\varphi_\epsilon(y):= -\epsilon \log\Bigg(\sum_j e^{-y_j/\epsilon}\Bigg); \quad \epsilon >0;
\end{align}
an approximation to the minimum. Indeed (see \cite{boyd2004convex}):
\begin{align}\label{eq.softmin}
	\min(y_j)-\epsilon \log(n)\leq \varphi_\epsilon(y)\leq \min(y_j).
\end{align}
$\varphi_\epsilon(y)$ is concave, and its gradient $\nabla \varphi_\epsilon(y)=:\delta(y)$ is an element of the unit simplex $\Delta_n\subset \R_+^n$, with components 
\begin{align}\label{eq.softargmin}
	\delta_{j}(y) = \frac{e^{-\frac{y_j}{\epsilon}}}{\sum_{k=1}^n e^{-\frac{y_k}{\epsilon}}};
\end{align}
as $\epsilon \to 0+$ these fractions concentrate their unit mass on the smallest coordinates of $y$. This leads us to consider  $y^i := \tau^i + \mu = (\tau_{ij}+\mu_j)_{j=1}^n$, the vector of delays-to-service seen from dispatcher $i$, and route according to the fractions\footnote{The notation emphasizes the dependence of the fractions on $\mu$, with the $\tau_{ij}$ assumed constant.} $\delta_{ij}(\mu) := \delta_j(\tau^i+\mu)$, which approximately follow the minimal delays, yielding the rates $x_{ij} = r_i \delta_{ij}(\mu)$. 

The full dynamics is:  
\begin{subequations}\label{eq.service dynamics}
	\begin{align}
		\dot{q}_j =& \sum_{i=1}^m x_{ij} - \min(q_j,c_j),\quad j=1,\ldots, n.\label{eq.dynamics-state}\\
		\mu_j(q_j) =& \left[\frac{q_j}{c_j}-1\right]^+,  \quad j=1,\ldots, n. \label{eq.dynamics-mu} \\
		x_{ij}=&r_i \delta_{ij}(\mu) = r_i\frac{e^{-\frac{\tau_{ij}+\mu_j}{\epsilon}}}{\sum_k e^{-\frac{\tau_{ik}+\mu_k}{\epsilon}}}, \quad
		\begin{array}{c} 
			i=1,\ldots, m, \\ j=1,\ldots, n. \end{array} 
		\label{eq.dynamics-x}
	\end{align}
\end{subequations}

Substituting \eqref{eq.dynamics-mu} into \eqref{eq.dynamics-x}, and subsequently into \eqref{eq.dynamics-state} yields an ordinary differential equation in the state variable $q=(q_j)_{j=1}^n \in \R_+^n$, the positive orthant. Note that  $q_j \geqslant 0$ is always preserved by the dynamics, since the negative drift term in \eqref{eq.dynamics-state} vanishes at $q_j=0$.

It is easily seen that the nonlinear terms in \eqref{eq.dynamics-state} and \eqref{eq.dynamics-mu} are globally Lipschitz over the entire domain. The same condition holds for \eqref{eq.dynamics-x} for fixed $\epsilon$, invoking a property of the soft-min: its Hessian $\nabla^2 \varphi_\epsilon(y) $ is a matrix with induced norm (maximum singular value) bounded by $1/\sqrt{\epsilon}$. Therefore, our dynamics \eqref{eq.service dynamics} has a globally Lipschitz field. As a consequence, given an initial condition $q(0)$, solutions to  \eqref{eq.service dynamics} exist, are unique, and defined for all times. 

\begin{remark}
	A similar dynamic model was analyzed in \cite{paganini2025tcns} for a very different application: selfish routing in EV charging infrastructures. Mathematically, the difference lies in the models for queue departure rates \eqref{eq.dynamics-state} and waiting times \eqref{eq.dynamics-mu}. Here, departures are contingent on service completion, as is natural in cloud systems; the dynamics in \cite{paganini2025tcns} considered sojourn-time based departures. The current departure model was considered in \cite{PagFAllerton23}, in its switching version, but only partially analyzed. This section provides a full analysis based on the soft-min approximation to the dynamics. 
\end{remark}	

\subsection{Optimization based analysis} \label{ssec.opt}

Consider the following optimization problem in the variables $X=(x_{ij})$:
\begin{subequations}\label{eq.opt}
	\begin{align}
		\min \sum_{i,j} \tau_{ij}x_{ij} &+ \epsilon \sum_{i,j}x_{ij}\log\left(\frac{x_{ij}}{r_i}\right) \label{eq.costfc}
		\\		\mbox{subject to:\quad}   x_{ij}&\geq 0  \ \forall i,j; 
		\quad \sum_j x_{ij}= r_i \ \forall i; \label{eq.demand constr}\\ 
		\sum_i x_{ij} &\leq c_j \ \ \forall j. \label{eq.supply constr}
	\end{align}
\end{subequations}

The constraints are natural: \eqref{eq.demand constr} is a restatement of \eqref{eq.xij}, and \eqref{eq.supply constr} says that no pool should be loaded  beyond its service capacity.

Regarding cost, since $x_{ij}$ is the rate in tasks/sec assigned by the dispatcher $i$ to the pool $j$, and $\tau_{ij}$ the setup time per task, then $\sum_{i,j}\tau_{ij}x_{ij}$ represents the number of tasks in the process of setup; a natural efficiency objective is to minimize this overhead. The second perturbation cost term (for small $\epsilon>0$) is a smoothing penalty, which makes the objective strictly convex. It may be expressed as: 
\begin{align}
	\label{eq.negentropycost} 
	\epsilon \sum_{i,j} r_i \delta_{ij} \log(\delta_{ij})
	= \epsilon\sum_i r_i \mathcal{H}(\delta^i),
\end{align}
where the negative entropy $\mathcal{H}(\delta^i) = \sum_j  \delta_{ij} \log(\delta_{ij})$
favors uniform routing among pools. 

We now characterize the feasibility conditions for our problem.

\begin{proposition}\label{prop.feas}
	The optimization problem \eqref{eq.opt} is feasible if and only if 	$\sum_i r_i \leq \sum_j c_j$. In that case, there is a unique optimum $X^*$. 
\end{proposition}
\begin{proof}
	Necessity of the condition follows immediately by adding constraints \eqref{eq.demand constr} over $i$, and invoking \eqref{eq.supply constr}:
	\[
	\sum_i r_i =\sum_{i,j} x_{ij} = \sum_j \sum_i x_{ij} \leq      
	\sum_j c_j. 
	\]
	For the converse, note that $X$ defined by
	\begin{equation}
		\label{eq.feasible x}
		x_{ij} = \frac{c_j}{\sum_k c_k} r_i
	\end{equation}
	always satisfies the constraints of \eqref{eq.opt} if $\sum_i r_i \leq \sum_j c_j$.
	
	Assuming feasibility, and noting that the cost function \eqref{eq.costfc} is strictly convex due to the entropy terms, we have a unique optimum $X^*$. 
\end{proof}

We will show that \eqref{eq.opt} characterizes the equilibrium of the dynamics \eqref{eq.service dynamics}. Write first the Lagrangian of the problem with respect to the constraints 
\eqref{eq.supply constr}, with multipliers $\mu_j \geq 0$:
\begin{equation}\label{eq.lagrangian}
	\begin{split}
		L(X,\mu) &= \sum_{i,j} \left[(\tau_{ij}+\mu_j)x_{ij} + \epsilon 
		x_{ij}\log\left(\frac{x_{ij}}{r_i}\right)\right] -\sum_j c_j \mu_j;
	\end{split}
\end{equation}
now minimize over $X$, under constraints \eqref{eq.demand constr} to find the dual function. We note that this minimization decouples over $i$; using the change of variables $x_{ij}=r_i\delta_{ij}$ it amounts to minimizing
\[
\sum_{j}(\tau_{ij}+\mu_j) \delta_{ij} + \epsilon \mathcal{H}(\delta^i)
\]
for each $i$, over $\delta^i = (\delta_{ij})_{j=1}^n \in \Delta_n$, the unit simplex. It is a standard exercise to show that the minimizer is $\delta(\tau^i+\mu)$ as in  \eqref{eq.dynamics-x}, and gives the minimum
\[
\varphi_\epsilon(\tau^{i}+\mu) =  -\epsilon\log\Bigg(\sum_j e^{-(\tau_{ij}+\mu_j)/\epsilon}\Bigg).
\]
Therefore the dual function is
\begin{align}\label{eq.dual}
	D(\mu) = \sum_i r_i \varphi_\epsilon(\tau^{i}+\mu)-\sum_j c_j \mu_j.
\end{align}

We now state our first main result:

\begin{theorem}\label{teo.service-eq} 
	Assume that $\sum_i{r_i} \leq \sum_j c_j$. The following are equivalent:
	\btmz
	\item[(i)] $(q^*,X^*,\mu^*)$ is an equilibrium point of 
	\eqref{eq.service dynamics}.   
	\item[(ii)] $(X^*,\mu^*)$ is a saddle point of the Lagrangian in \eqref{eq.lagrangian},  and $q^*$ is given by:
	\begin{align}\label{eq.defqstar}
		q^*_j = \sum_{i}x^*_{ij}+c_j\mu_j^* = \begin{cases} 
			c_j \left(1+ \mu_j^*\right)  & \mbox{ if }\mu^*_j>0; \\
			\sum_ix_{ij}^* & \mbox{ if }\mu^*_j=0.
		\end{cases}
	\end{align}
	\etmz
\end{theorem}

\vpo

\begin{proof}
	Starting with (i), we verify the saddle point conditions. 
	In the preceding discussion leading up to the dual function, we already established that
	\begin{align}\label{eq.saddle X}
		X^* \in \arg\min_X L(X,\mu^*) \ \Sii \ & x^*_{ij} = r_i \delta_{ij}(\mu^*), \quad \mbox{ with }\delta_{ij}(\mu) \mbox{ in \eqref{eq.dynamics-x}}. 
	\end{align}	
	Also, the equilibrium implies that
	\begin{align}\label{eq.feas equil}
		\sum_{i} x^*_{ij} = \min(q^*_j,c_j)\leq c_j,
	\end{align}
	so $X^*$ is primal-feasible. Dual feasibility ($\mu^* \geq 0$) follows by 
	\eqref{eq.dynamics-mu}. 
	
	For complementary slackness, note that if $\mu_j^* >0$, then \eqref{eq.dynamics-mu} yields $q_j^*=c_j(1+\mu^*_j) > c_j$. Therefore, $\min(q^*_j,c_j)= c_j$ and $\sum_{i} x^*_{ij} = c_j$ in \eqref{eq.feas equil}, the constraint is active as required. We have also established the first case of \eqref{eq.defqstar}.
	
	If $\mu_j^* = 0$, then \eqref{eq.dynamics-mu} yields $q_j^*\leq c_j$, so $\min(q^*_j,c_j)= q^*_j$ and $\sum_{i} x^*_{ij} = q^*_j$, the second case of \eqref{eq.defqstar}.
	
	Start now with (ii). We claim that at each $j$ both of the following identities hold:  
	\begin{subequations}\label{eq.two conditions}
		\begin{align}
			\mu^*_j &= \left[\frac{q_j^*}{c_j} - 1 \right]^+;  \label{eq.mustar}\\
			\sum_i x^*_{ij} &= \min(q^*_j,c_j).\label{eq.qhat}
		\end{align}
	\end{subequations}
	\btmz
	\item If $\mu_j^*>0$, then $q^*_j =c_j(1+\mu^*_j)$ from \eqref{eq.defqstar}, which gives \eqref{eq.mustar}. Also, $q^*_j > c_j$, and $\sum_i x^*_{ij} = c_j$ due to complementary slackness, which implies \eqref{eq.qhat}. 
	\item If $\mu_j^*=0$, then $q^*_j = \sum_i x^*_{ij} \leq c_j$ by \eqref{eq.defqstar} and primal feasibility. So, \eqref{eq.mustar}-\eqref{eq.qhat} are valid here as well. 
	\etmz
	
	Now \eqref{eq.mustar} implies that $(\mu^*,q^*)$ are consistent with \eqref{eq.dynamics-mu}, and \eqref{eq.qhat} implies that $(X^*, q^*)$ is an equilibrium of \eqref{eq.dynamics-state}. 
	
\end{proof}

Regarding the question of uniqueness of the equilibrium: note that $X^*$ must be unique because there is a unique primal optimal point for \eqref{eq.opt}. There is, however, a situation with non-unique $\mu^*$ (and $q^*$), now succinctly described. 

Assume an equilibrium with \emph{all} pools saturated at rate $c_j$ (this requires
$\sum_i{r_i} = \sum_j c_j$). Given a valid $\mu^*$, adding a positive constant $K$ to all 
$\mu^*_j$ does not modify the rate allocation according to \eqref{eq.dynamics-x}; if we also increase the $q^*_j$ accordingly to satisfy \eqref{eq.dynamics-mu}, then the equilibrium condition in \eqref{eq.dynamics-state} is maintained since both $\sum_i x^*_{ij}$ and 
$\min(q^*_j,c_j)$ are unchanged.  

Eliminating this border case, we can state the following:

\begin{proposition} \label{prop.unique}
	Assume that $\sum_i{r_i} < \sum_j c_j$. Then there is a unique equilibrium point 
	$(q^*,X^*,\mu^*)$ of \eqref{eq.service dynamics}.  	
\end{proposition}

\begin{proof}
	We use the characterization of Theorem \ref{teo.service-eq}. Assume that the set $\arg\max_{\mu \in \R_+^n} D(\mu)$ is not a singleton. Since the function is concave, there must be a segment $\mu^* + tv$, $0\neq v\in \R^n$, $0\leq t < \eta$ where $D(\mu^* + tv)$ is constant. 
	
	By double differentiation and noting from \eqref{eq.dual} that
	\[
	\nabla^2 D(\mu) = \sum_i r_i \nabla^2 \varphi_\epsilon(\tau^{i}+\mu), 
	\] 
	we must have $v \in \ker\nabla^2 \varphi_\epsilon(\tau^{i}+\mu^*)$ for each $i$. The 
	Hessian of a log-sum-exp function (see\cite{boyd2004convex}) has kernel spanned by  $v=\mathbf{1}$.
	Furthermore, it is easily checked  that $\varphi_\epsilon(\tau^{i}+\mu + t \mathbf{1})=
	\varphi_\epsilon(\tau^{i}+\mu) + t$. Therefore, we have
	\[
	D(\mu^*+t \mathbf{1}) = D(\mu^*)+ t\Big(\sum_i r_i -\sum_j c_j \Big),
	\]
	a contradiction, since the function was constant but the slope coefficient is negative. Thus, 
	$\mu^* = \arg\max_{\mu \in \R_+^n} D(\mu)$ is unique, and then $q^*$ satisfying \eqref{eq.defqstar} is also unique.
\end{proof} 

To complete the characterization, we establish the global attractiveness of the equilibrium. 

\begin{theorem}\label{teo.convergence}
	Assume that $\sum_i{r_i} < \sum_j c_j$. Any trajectory $(q(t),X(t),\mu(t))$ of \eqref{eq.service dynamics} converges asymptotically to the unique equilibrium point $(X^*,q^*,\mu^*)$.    
\end{theorem}
Proof is given in the Appendix.

The preceding theory suggests that a myopic dispatch based on the delay-to-service criterion may be a good control strategy: it reaches rates $X^*$ which are socially optimal in the sense of minimizing the number of tasks in setup, a reasonable criterion. 
Two limitations must, however, be pointed out:
\btmz
\item Our equilibrium includes an undesirable cost of the allocation, because in addition to setup it must tolerate some \emph{waiting} at server pools, for every $j: q^*_j > c_j$, i.e., every saturated pool. Such waiting would not be inevitable for a central planner with global information, who could compute $X^*$ offline. If such entity could impose the equilibrium dispatch, pool input rates would never exceed $c_j$ and no queue buildup would arise. The 
inefficiency of the myopic solution is a result of the signal used to measure congestion, namely queueing delays. Such ``Price of Anarchy" is a common occurrence in applications such as transportation where routing is selfish \cite{roughgarden2005selfish,PagFAllerton23}, only responding to individual (per job) incentives. Here we are free to use a more convenient signal, provided that the decentralized architecture of dispatchers is preserved. One such proposal is considered in the next section.
\item As mentioned before, our dynamic model \eqref{eq.dynamics-state} involves an approximation, having neglected setup delays in the arrivals to the queues. A first observation is that the \emph{equilibrium} characterization would still apply to a model such as \eqref{eq.delaydiff}. But the convergence result is not valid, and oscillations could appear around the optimal allocation. The following section addresses this limitation as well.  
\etmz

\section{Fluid model with setup queues} \label{sec.setup}

In this section we consider an alternative dynamic model, where the dynamics of tasks under setup is explicitly considered. We denote by $z_{ij}$ the population of tasks from dispatcher $i$ which are undergoing setup for pool $j$, and express its evolution by the fluid queue:
\begin{align}\label{eq.zij}
	\dot{z}_{ij}= x_{ij} - \gamma_{ij} z_{ij}.
\end{align}
The rationale is as follows: tasks enter setup at a rate $x_{ij}$, driven by the dispatcher decisions. Each task leaves setup at a rate $\gamma_{ij}:= 1/\tau_{ij}$, inverse of the mean setup time. Thus $\gamma_{ij} z_{ij}$ denotes the rate at which setup phases  conclude and the tasks feed the queue at pool $j$; consistently, the queue dynamics must be modified to 
\begin{align}\label{eq.queue with setup}
	\dot{q}_{j} =  \sum_{i=1}^m \gamma_{ij} z_{ij} - \min(q_j,c_j), \quad j = 1,\ldots, n.
\end{align}

From a transfer function perspective, we are interposing a first-order lag
\[
\frac{\gamma_{ij}}{s+\gamma_{ij}} = \frac{1}{\tau_{ij} s + 1}
\] 
between arrival to the dispatcher and arrival to the queue, instead of the pure delay $e^{-\tau_{ij}s}$ implicit in \eqref{eq.delaydiff}. This can be justified as a simpler model for analysis, and also as a means to capture variability in the setup process: in queueing theory terms, \eqref{eq.zij} is the fluid model corresponding to an $M/M/\infty$ queue with Poisson arrival rate $x_{ij}$ and random, $\exp(\gamma_{ij})$-distributed setup times.  	

We will develop an appropriate dispatcher control for the system \eqref{eq.zij}-\eqref{eq.queue with setup}, with these conditions/objectives:
\btmz
\item Dispatcher $i$ may receive information from all pools $j$ as before, and also from the setup queue of its own tasks, but not from other dispatchers $i'\neq i$. With this information, they must control $x_{ij}$.
\item In equilibrium, we seek to achieve the same optimal allocation as before, preferably with no price of anarchy in terms of waiting queues. 
\etmz

\section{Balancing based on proximal optimization}\label{sec.prox}

In Section \ref{sec.myopic} we proposed the myopic routing rule first, and then introduced the optimization problem \eqref{eq.opt} to analyze it; the queue dynamics was strongly  connected with the dual of this convex program. 

In this section we will formulate an optimization problem first, with additional variables that are aimed at representing the setup queues. From the optimization we will derive a control law for routing that enables a strong connection with the expanded dynamics  \eqref{eq.zij}-\eqref{eq.queue with setup}.  For this purpose we will rely on a proximal regularization term, together with a primal-dual gradient dynamics, as studied in \cite{goldsztajn2021proximal}.

Introduce the following optimization problem in the variables $X = (x_{ij})$ and $Z = (z_{ij})$:
\begin{subequations}\label{pr.proximal}
	\begin{align}
		\min_{X, Z} &\double \sum_{i, j} \left[\frac{x_{ij}}{\gamma_{ij}} + \frac{(x_{ij} - \gamma_{ij}z_{ij})^2}{2\gamma_{ij}}\right] \label{pr1.proximal} \\
		\text{subject to:} &\double x_{ij} \geq 0 \double \forall i,j; \double \sum_j x_{ij}= r_i \double \forall i; \label{eq.proximal demand} \\ 
		&\double \sum_i x_{ij} \leq c_j \double \forall j. \label{eq.proximal supply}
	\end{align}
\end{subequations}
We observe that the constraints in $X$ are identical to \eqref{eq.demand constr}-\eqref{eq.supply constr}. The only change (note $\gamma_{ij}=\tau_{ij}^{-1}$) is in the second term of the objective function, where we replace the entropy regularization with a quadratic penalty involving the additional variable $Z$. Since the latter is unconstrained, clearly the penalty disappears at optimality, and the optimal cost will be reduced to the first term in \eqref{eq.costfc}. Therefore, we are aiming for the same rate allocation as in the previous section, in the limit of small $\epsilon>0$.   

\begin{remark}
	Since the constraints are identical, it follows from Proposition \ref{prop.feas} that Problem \eqref{pr.proximal} is feasible and has a finite optimum if and only if $\sum_i r_i \leq \sum_j c_j$.
\end{remark}
%

\subsection{Reduced Lagrangian}
\label{sub: reduced lagrangian}

To analyze the new optimization problem with the tools of \cite{goldsztajn2021proximal}, it is convenient to introduce the notation $\Delta_r \defequal \sset{X \in \R_+^{m \times n}}{\sum_j x_{ij} = r_i\ \forall i}$, 
and the convex indicator function $\map{\chi_{\Delta_r}}{\R^{m \times n}}{\{0, +\infty\}}$  defined by:
\begin{equation*}
	\chi_{\Delta_r}(X) = 0 \double \text{if} \double X \in \Delta_r, \double \chi_{\Delta_r}(X) = +\infty \double \text{if} \double X \notin \Delta_r;
\end{equation*}
adding this function to the cost in \eqref{pr1.proximal} automatically enforces constraints  \eqref{eq.proximal demand}; we treat the remaining constraints 
\eqref{eq.proximal supply} by Lagrange duality. The Lagrangian is:
\begin{equation*}
	\begin{split}
		L(X, Z, \nu) &\defequal \sum_{i, j} \left[\frac{x_{ij}}{\gamma_{ij}} + \frac{\left(x_{ij} - \gamma_{ij}z_{ij}\right)^2}{2\gamma_{ij}}\right] + \sum_j \nu_j \left(\sum_i x_{ij} - c_j\right) + \chi_{\Delta_r}(X),
	\end{split}
\end{equation*}
defined in $\R^{m \times n} \times \R^{m \times n} \times \R^n$. The strategy of 
\cite{goldsztajn2021proximal} for analysis is to compute the \emph{reduced } Lagrangian which results from partial minimization over $X$, namely: 
\begin{equation*}
	\bL(Z, \nu) \defequal \min_X L(X, Z, \nu), \double (Z, \nu) \in \R^{m \times n} \times \R^n,
\end{equation*}
achieved at 
\begin{equation}\label{eq.xbar}
	\bX(Z, \nu) \defequal \argmin{X} L(X, Z, \nu), \double (Z, \nu) \in \R^{m \times n} \times \R^n.
\end{equation}
We note that $\bX(Z, \nu)$ is well defined due to the strong convexity of $L(X, Z, \nu)$  in $X$. 
%

\begin{remark} \label{rem.tacprox}
	The proximal term used in \cite{goldsztajn2021proximal} is the square norm of $X-Z$, which would correspond to setting $\gamma_{ij} = 1$ in the second term of \eqref{pr1.proximal}. The present generalization, where a weight is included, requires a minor adaptation of the theory in \cite{goldsztajn2021proximal}. We omit the details. 
\end{remark}

$\bL(Z, \nu)$ is convex in $Z$ and concave in $\nu$, and it follows from similar arguments as in \cite{goldsztajn2021proximal} that  $\bX$ is locally Lipschitz. Moreover, $(\hat{X}, \hat{Z}, \hat{\nu})$ is a saddle point of $L$, i.e.,
\begin{equation*}
	L(\hat{X}, \hat{Z}, \nu) \leq L(\hat{X}, \hat{Z}, \hat{\nu}) \leq L(X, Z, \hat{\nu})
\end{equation*}
if and only if $(\hat{Z}, \hat{\nu})$ is a saddle point of $\bL$, i.e.,
\begin{equation*}
	\bL(\hat{Z}, \nu) \leq \bL(\hat{Z}, \hat{\nu}) \leq \bL(Z, \hat{\nu}), \double \text{and} \double \bX(\hat{Z}, \hat{\nu}) = \hat{Z}.
\end{equation*}

The latter properties imply that $\bL$ can be used to solve \eqref{pr.proximal} instead of $L$, with the following advantage:


\begin{proposition}
	\label{prop.gradients}
	$\bL$ is differentiable, with
	\begin{subequations}
		\label{eq.gradients}
		\begin{align}
			&\frac{\partial \bL}{\partial z_{ij}}(Z, \nu) = \gamma_{ij}z_{ij} - \bx_{ij}(Z, \nu); \label{eq1.gradients} \\
			&\frac{\partial \bL}{\partial \nu_j}(Z, \nu) = \sum_i \bx_{ij}(Z, \nu) - c_j; \label{eq2.gradients}
		\end{align}
	\end{subequations}
	where we use the notation $\bX = (\bx_{ij})$.
\end{proposition}

\subsection{Saddle-point gradient dynamics}
\label{sub: saddle-point gradient dynamics}

A classical strategy \cite{arrow1958studies} to find a saddle point of a convex-concave function is the gradient dynamics:
\begin{subequations}
	\label{eq.saddle}
	\begin{align}
		&\dot{z}_{ij} = -\frac{\partial \bL}{\partial z_{ij}}(Z, \nu) = \bx_{ij}(Z, \nu) - \gamma_{ij}z_{ij}, \label{eq1.saddle} \\
		&\dot{\nu}_j = \left[\frac{\partial \bL}{\partial \nu_j}(Z, \nu)\right]_{\nu_j}^+ = \left[\sum_i \bx_{ij}(Z, \nu) - c_j\right]_{\nu_j}^+. \label{eq2.saddle}
	\end{align}
\end{subequations}
Here $[\alpha]_\beta^+ \defequal \alpha$ if $\alpha > 0$ or $\beta > 0$, and $[\alpha]_\beta^+ = 0$ otherwise, ensuring that the multipliers $\nu_j$ remain non-negative. 

It is easy to check that $(Z^*, \nu^*)$ is an equilibrium point of \eqref{eq.saddle} if and only if $(Z^*, \nu^*)$ is a saddle point of $\bL$. We will assume that such a saddle point exists, which is always the case when problem \eqref{pr.proximal} is strictly feasible.
\begin{remark}
	\label{rem.strict feasibility}
	It is straightforward to check that
	\begin{equation*}
		\sum_i r_i < \sum_j c_j, \double r_i > 0 \double \forall i, \double c_j > 0 \double \forall j
	\end{equation*}
	imply that $X$ as in \eqref{eq.feasible x} is strictly feasible for \eqref{pr.proximal}.
\end{remark}

The following results concern the global stability of the dynamics \eqref{eq.saddle} and follow from \cite{goldsztajn2021proximal}, again with small variations as mentioned in Remark \ref{rem.tacprox}:
\begin{itemize}
	\item If $(\hat{Z}, \hat{\nu})$ is any saddle point of $\bL$, then
	\begin{equation*}
		V(Z, \nu) \defequal \frac{1}{2}\norm{Z - \hat{Z}}_2^2 + \frac{1}{2}\norm{\nu - \hat{\nu}}_2^2
	\end{equation*}
	is nonincreasing along the trajectories of \eqref{eq.saddle}, i.e.,
	\begin{equation*}
		\dot{V} = \nabla V \left(\nabla \bL\right)^\intercal \leq 0.
	\end{equation*}
	
	\item Each solution of \eqref{eq.saddle} converges to an equilibrium point of \eqref{eq.saddle} contained in $\sset{(Z, \nu)}{\dot{V}(Z, \nu) = 0}$; recall that all equilibria are saddle points of $\bL$ and thus solve \eqref{pr.proximal}.
\end{itemize}

\subsection{Proximal routing rule}

The saddle-point dynamics \eqref{eq.saddle} was drawn from the theoretical work in \cite{goldsztajn2021proximal}, a priori not related to the load balancing dynamics. We observe, however, that \eqref{eq1.saddle} is identical to \eqref{eq.zij}, provided that we use the rule $\bar{x}_{ij}(Z,\nu)$ as a criterion for routing. We verify that this rule satisfies our information requirements. The minimization  
in \eqref{eq.xbar} involves the constraints $X\in \Delta_r$, and the cost terms
\begin{equation}\label{eq.proximal routing}
	\sum_i \sum_j\left[\left(\frac{1}{\gamma_{ij}} + \nu_j\right)x_{ij} + \frac{\left(x_{ij} - \gamma_{ij}z_{ij}\right)^2}{2\gamma_{ij}}\right].
\end{equation}
This indeed decouples over the dispatchers. For each $i$ it amounts to a quadratic program, minimizing its portion of the cost \eqref{eq.proximal routing} subject to the constraints \eqref{eq.proximal demand}. In fact, an equivalent finite-step algorithm which projects on the simplex can be given, similarly to what is done in \cite{goldsztajn2019proximal}. In terms of information, dispatcher $i$ must know the state of multipliers $\nu_j$, and setup-queue lengths $z_{ij}$ for type $i$ jobs. 

\subsection{Multipliers via virtual queues}

In contrast with the setup queue, there is \emph{not} a direct coincidence between the dynamics \eqref{eq.queue with setup} of the pool queue and the respective multiplier dynamics in  \eqref{eq2.saddle}. The latter behaves as a \emph{virtual} queue which is:
\begin{itemize}
	\item Incremented when any dispatcher $i$ assigns a new task to the pool, without waiting for the setup time to occur;
	\item Decremented at the full rate $c_j$, irrespective of the number of jobs present in the pool. 
\end{itemize}
This scheme is more involved than the myopic one in Section \ref{sec.myopic}, but note that updates may be carried out by each pool without the need for any global information. 

Running this kind of dynamics would ensure an equilibrium in which optimal assignment rates $x^*_{ij}$ are reached asymptotically. However, since queues $q_j$ are not explicitly present in the feedback loop, their limiting behavior is not immediate. In particular, the queueing delay inefficiency pointed out in the previous section is not easily ascertained, it would in general depend on initial conditions. 

A solution is to slightly modify the problem to guarantee the  desired equilibrium properties for $q_j$. This may be done by replacing $c_j$ in \eqref{eq2.saddle} with tunable parameters $\tilde{c}_j < c_j$. These must be chosen so $\sum_i r_i \leq \sum_j \tilde{c}_j$, and thus Problem \eqref{pr.proximal} remains feasible with a slightly tightened feasibility region. Assuming the change is small, the new solution will have almost optimal cost. 

Based on this idea, consider the following combined dynamics: 
\begin{subequations}
	\label{eq.queues and saddle}
	\begin{align}
		&\dot{q}_j = \sum_i \gamma_{ij}z_{ij} - \min\left(q_j, c_j\right), \label{eq1.queues and saddle} \\
		&\dot{z}_{ij} = \bx_{ij}(Z, \nu) - \gamma_{ij}z_{ij}, \label{eq2.queues and saddle} \\
		&\dot{\nu}_j = \left[\sum_i \bx_{ij}(Z, \nu) - \tilde{c}_j\right]_{\nu_j}^+. \label{eq3.queues and saddle} 
	\end{align}
\end{subequations}
We analyze the properties of its solution.

\begin{theorem}
	\label{the.equilibria of proximal dynamics}
	Suppose that
	\begin{equation*}
		\sum_i r_i < \sum_j \tilde{c}_j, \double r_i > 0 \double \forall i, \double c_j > \tilde{c}_j > 0 \double \forall j. 
	\end{equation*}
	Then, solutions of \eqref{eq.queues and saddle} converge to an equilibrium point $(q^*, Z^*, \nu^*)$ such that: 
	\begin{enumerate}
		\item[(a)] $(\bX(Z^*, \nu^*), Z^*, \nu^*)$ solves \eqref{pr.proximal} with $\tilde{c}_j$ replacing $c_j$,
		
		\item[(b)] $q_j^* = \sum_i \gamma_{ij} z_{ij}^* < c_j$ for all $j$. 
	\end{enumerate}
\end{theorem}

\begin{proof}
	First, note that Remark \ref{rem.strict feasibility} implies that \eqref{pr.proximal} is strictly feasible if $c_j$ is replaced by $\tilde{c}_j$, so the results stated in Sections \ref{sub: reduced lagrangian} and \ref{sub: saddle-point gradient dynamics} hold, replacing $c_j$ by $\tilde{c}_j$.
	
	Fix any solution $(q(t), Z(t), \nu(t))$ of \eqref{eq.queues and saddle} and observe that \eqref{eq2.queues and saddle}-\eqref{eq3.queues and saddle} are independent of \eqref{eq1.queues and saddle}. Hence, $(Z(t), \nu(t))$ converges as $t \to \infty$ to a point $(Z^*, \nu^*)$ that is an equilibrium of \eqref{eq2.queues and saddle}-\eqref{eq3.queues and saddle} and a saddle of $\bL$; this follows from the results at the end of Section \ref{sub: saddle-point gradient dynamics}. In particular, (a) holds.
	
	Because $(Z^*, \nu^*)$ is an equilibrium point of \eqref{eq2.queues and saddle}-\eqref{eq3.queues and saddle},
	\begin{equation*}
		\sum_i \gamma_{ij}z_{ij}^* = \sum_i \bx_{ij}(Z^*, \nu^*) \leq \tilde{c}_j < c_j, \double \forall j.
	\end{equation*}
	Therefore, there exist $t_0 \geq 0$ and $0 < \varepsilon_j < c_j - \tilde{c}_j$ such that
	\begin{equation*}
		\dot{q}_j(t) \leq \tilde{c}_j + \varepsilon_j - \min(q_j(t), c_j) \double \forall j, \double \forall t \geq t_0, 
	\end{equation*}
	which implies that there exists $t_1 \geq t_0$ such that $q_j(t) \leq c_j$ for all $j$ and $t \geq t_1$. Then we obtain
	\begin{equation*}
		\dot{q}_j(t) = \sum_i \gamma_{ij}z_{ij}(t) - q_j(t) \double \forall j, \double \forall t \geq t_1.
	\end{equation*}
	This equation represents a linear, first-order system driven by an input which converges as $t\to \infty$. We conclude that $q(t) \to q^*$ as $t \to \infty$ with $q^*$ given by the expression in (b).
\end{proof}

To summarize: we have found a procedure, the proximal routing rule with multipliers generated by virtual queues, that produces converging load balancing rates $\bx_{ij}(t)$.
By choosing $\tilde{c}_j$ close to $c_j$, the limit rates approximate the solution of Problem \eqref{pr.proximal}, where the number of tasks $z^*_{ij} = \gamma_{ij}^{-1}  x^*_{ij}$ undergoing setup is minimized. Moreover, (b) implies that there is no waiting at pools in equilibrium: all the tasks that have finished their setup are immediately served. 

We have thus addressed both limitations pointed out at the end of Section \ref{sec.myopic}.

\section{Numerical experiments}
\label{sec.simulations}
We present simulation experiments that illustrate the results of Sections \ref{sec.myopic} and \ref{sec.prox}. We consider a small setup with $m = n = 2$ and the following parameters:
\begin{equation*}
	c = \left[\begin{matrix}
		15 \\
		10
	\end{matrix}\right], \double
	r = \left[\begin{matrix}
		16 \\
		8
	\end{matrix}\right], \double
	\tau = \left[\begin{matrix}
		1 & 2 \\
		2 & 1
	\end{matrix}\right];
\end{equation*}
recall that $\gamma_{ij} = 1 / \tau_{ij}$. If setup times were the only factor, we would send all tasks of type $i$ to server pool $j = i$. However, this violates the capacity constraint \eqref{eq.supply constr}; if applied, the queue at pool $1$ would grow without limit.  
Hence, a load balancing rule of the kind discussed in this paper is required. 

First consider the myopic routing rule of Section \ref{sec.myopic}. For $\epsilon = 0$, the minimizer of \eqref{eq.opt} is:
\begin{equation}
	\label{eq.solution 1}
	x_{11}^* = 15, \double x_{12}^* = 1, \double x_{21}^* = 0, \double x_{22}^* = 8;
\end{equation}
it does not change appreciably if $\epsilon = 0.01$. For this $\epsilon$, we solved \eqref{eq.service dynamics} numerically when the queues are empty at time zero. The routing rates are shown in Figure \ref{fig: rates}, and the queue lengths in Figure \ref{fig: queues}. The routing rates converge to the equilibrium in \eqref{eq.solution 1}, which minimizes the number of tasks in the process of setup. However, the equilibrium queue at server pool~$1$ almost doubles its capacity; all those tasks incur a significant wait before being processed.

\begin{figure}
	\centering
	\begin{subfigure}{0.49\columnwidth}
		\centering
		\includegraphics[width=\columnwidth]{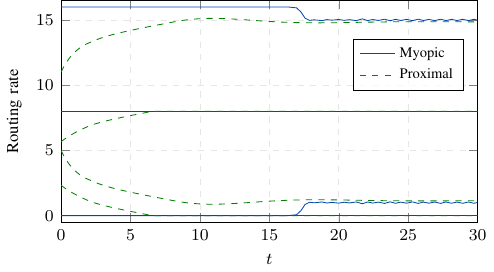}
		\subcaption{Routing rates}
		\label{fig: rates}
	\end{subfigure}
	\hfill
	\begin{subfigure}{0.49\columnwidth}
		\centering
		\includegraphics[width=\columnwidth]{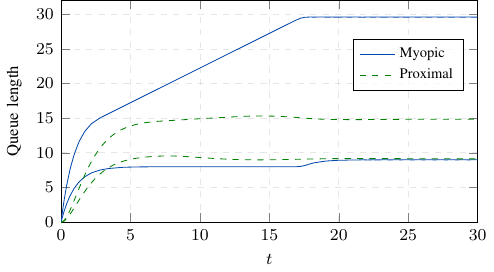}
		\subcaption{Queue lengths}
		\label{fig: queues}
	\end{subfigure}
	\caption{Evolution of the routing rates and number of tasks at each server pool for both policies. To identify the routing rates on the left panel, note that at equilibrium they satisfy $x_{11} > x_{22} > x_{12} > x_{21}$ for both policies. On the right panel, note that $q_1 > q_2$ in equilibrium for both policies.}
\end{figure}

Consider now the proximal routing rule of Section \ref{sec.prox}. For $\epsilon \approx 0$, the minimizer for \eqref{pr.proximal} is again \eqref{eq.solution 1}.  If, instead, we replace $c_j$ by $\tilde{c}_j = 0.99 c_j$, the new minimizer of \eqref{pr.proximal} is found to be:
\begin{equation*}
	x_{11}^* = 14.85, \double x_{12}^* = 1.15, \double x_{21}^* = 0, \double x_{22}^* = 8.
\end{equation*}
Also, the optimal cost increases from $25$ to $25.15$. This slight increase in the number of tasks undergoing setup is amply compensated by a reduction of queue lengths, as shown in Figures \ref{fig: rates} and~\ref{fig: queues}. In particular, $q_j < c_j$ for all $j$ in equilibrium, as proved in Theorem~\ref{the.equilibria of proximal dynamics}. This implies that, in equilibrium, all tasks are processed immediately after their setup concludes.

\section{Conclusion} \label{sec.concl}

We have studied load balancing over a set of server pools, carried out by a set of dispatchers, one per task type. The heterogeneous types are distinguished by the setup times each task requires to enter service at each pool. 

Two decentralized control strategies were proposed, based on feedback from the pool queues and 
the type-specific setup process. Both are proved to reach equilibrium rates consistent with the minimization of overall setup delay. The simpler myopic rule achieves this at the cost of queueing delay; the more expensive proximal routing rule combined with virtual queueing achieves all desired equilibrium features.


\begin{appendices}
	
\section{Proof of Theorem \ref{teo.convergence}}\label{app.conv}

\begin{proof}
Our proof of global convergence of the dynamics \eqref{eq.service dynamics} is based on using the dual function $V(q)=D(\mu(q))$ as a Lyapunov function, and a LaSalle-type argument. As such, it resembles the proof given in \cite{paganini2025tcns} for another dynamics which has a similar structure, but also relevant differences. 

We first note that  the dual function $D(\mu)$ must be upper bounded for $\mu\geq 0$ if the optimization problem is feasible; we will require a sharper bound for the case of strict feasibility. Assume $\sum_i r_i = (1-\eta) \sum_j c_j$, for some $\eta > 0.$
Invoking \eqref{eq.dual} we have:
\begin{align*}
	D(\mu )	&= \sum_i r_i \varphi_\epsilon(\tau^{i}+\mu) - \sum_j \mu_j c_j \\
	&\leq \sum_i r_i \min_j(\tau_{ij}+\mu_j) - \sum_j \mu_j c_j \\
	&\leq  \left[ \bar{\tau} +\min_j\mu_j \right]\sum_i r_i - \sum_j \mu_j c_j  	\\
	&=  \bar{\tau}\sum_i r_i  +\min_j\mu_j (1-\eta) \sum_j c_j 	- \sum_j \mu_j c_j 
	\\
	&\leq   \bar{\tau}\sum_i r_i +(1-\eta) \sum_j \mu_j c_j- \sum_j \mu_j c_j  \leq   \bar{\tau}\sum_i r_i -\eta \sum_j \mu_j c_j. 
\end{align*}
Here the first step uses \eqref{eq.softmin}, for the second step we take $\bar{\tau}$ to be  an upper bound on all setup times. A consequence of this bound is that for any $\mu_0 \geq 0$, the set 
\[
\mathcal{M}:= \{\mu\geq 0: D(\mu)\geq D(\mu_0)\}
\]
is bounded in $\R_+^n$. Given \eqref{eq.mu}, this also implies that 
\begin{align}
	\label{eq.set Q}
	\mathcal{Q}:= \{q\geq 0: D(\mu(q))\geq D(\mu_0)\}
\end{align}
is a compact set in  $\R_+^n$.

Consider now any trajectory $q(t)$ of the dynamics \eqref{eq.service dynamics}, continuously differentiable. $\mu_j(q_j)$ is not differentiable at $q_j=c_j$,  nevertheless it is Lipschitz, so  $\mu(q(t))$ is absolutely continuous: time derivatives exist almost everywhere,  and the function is an integral of its derivative. 

Let $\mathcal{T}\subset \R_+$ (with complement of Lebesgue measure zero) be the set of times for which $\dot{\mu}(t)$ exists. For such times, note from \eqref{eq.dynamics-mu} that $\dot{\mu_j}(t)=0$ when  $q_j(t)<c_j$; furthermore, for $\dot{\mu_j}(t)$ to exist simultaneously with $q_j(t)=c_j$ it is necessary that $\dot{\mu_j}(t)=0$. Finally, if $q_j(t) > c_j$, then $\dot{\mu_j} = \frac{\dot{q}_j}{c_j}$. 

We now compose $\mu(t)$ with the dual function $D(\mu)$ which is continuously differentiable, with partial derivatives:
\begin{align}
	\label{eq.partial D}
	\frac{\partial D}{\partial \mu_j} = \sum_i r_i \delta_{ij}(\mu) - c_j.
\end{align}
At points $t\in \mathcal{T}$ we can apply the chain rule to obtain 
\begin{align}
	\frac{d}{dt}{D}(\mu(q(t))) & = \sum_j\frac{\partial D}{\partial \mu_j}\dot{\mu}_j \nonu \\
	& =\sum_{j:q_j(t) > c_j} \left[\sum_i x_{ij}(t) - c_j\right] \frac{\dot{q}_j}{c_j} =  \sum_{j:q_j(t) > c_j} \frac{(\dot{q}_j)^2}{c_j}\geq 0. \label{eq.ddot}
\end{align} 
In the second step above, we incorporated \eqref{eq.dynamics-x} into \eqref{eq.partial D}, and substituted for 
$\dot{\mu}_j$ removing terms where it is zero. Finally we noticed that for the remaining terms the square bracket is precisely the right-hand side of  \eqref{eq.dynamics-state}.

By integration over time we conclude that $V(q) = D(\mu(q))$ is non-decreasing along any trajectory $q(t)$ of the dynamics \eqref{eq.service dynamics}. This implies, in turn, that 
$\mathcal{Q}$ in \eqref{eq.set Q} is a compact invariant set of the dynamics, for any $\mu_0$. 

Before proceeding further, we state a modified version of the queue dynamic equation which results from noticing that
\[
q_j = \min(q_j,c_j)+[q_j - c_j]^+ = \min(q_j,c_j)+ c_j \mu_j,
\]
where we used \eqref{eq.dynamics-mu}, and then substituting into  \eqref{eq.dynamics-state}:
\begin{align}\label{eq.first order q}
	\dot{q}_j =  - q_j + \sum_{i} x_{ij}(\mu(t))+c_j\mu_j(t).
\end{align}

Now consider an arbitrary trajectory $q(t)$ of the dynamics; we must show it converges to equilibrium. Setting $\mu_0 =\mu(q(0))$ we have $q(t) \in \mathcal{Q}$, compact; let $L^+\subset \mathcal{Q}$ be its $\Omega$-limit set, itself also invariant under the dynamics. We know that  $V(q(t))$ is bounded and monotonically non-decreasing, let its limit be $\bar{V}$. By continuity, $V(q)\equiv \bar{V}$ for any $q\in L^+$. 

Now consider an auxiliary trajectory, $\tilde{q}(t)$, with initial condition $q^+ \in L^+$. We conclude that 
\[
V(\tilde{q}(t))=D(\mu(\tilde{q}(t)))\equiv \bar{V} \ \Longrightarrow\ \frac{d}{dt}D (\mu(\tilde{q}(t)))\equiv 0.
\]
In reference to \eqref{eq.ddot}, this implies that, almost everywhere, we must have 
$\dot{\tilde{q}}_j =0$ for any $j: \tilde{q}_j(t)>c_j$, and therefore  $\dot{\tilde{\mu}}_j= 0$ for such queues. But non-congested queues ($\tilde{q}_j(t)\leq c_j$) also satisfy  $\dot{\tilde{\mu}}_j= 0$ as already mentioned, so $\dot{\tilde{\mu}}\equiv 0$ and  $\mu(\tilde{q}(t))\equiv \tilde{\mu}$, constant. Consequently, the rates $X(\tilde{\mu})$ defined by \eqref{eq.dynamics-x} are also constant. 

Using \eqref{eq.first order q} for the trajectory $\tilde{q}$, we note that this is a \emph{linear} first-order system with a {constant} input $\sum_{i} x_{ij}(\tilde{\mu})+c_j\tilde{\mu}_j$. Therefore the solutions $\tilde{q}_j(t)$
converge to this value, an equilibrium of the dynamics. However, since the  equilibrium is unique from Theorem \ref{teo.service-eq}, we conclude from this analysis of $\tilde{q}(t)$ that  necessarily $\tilde{\mu} = \mu^*$, the dual optimal price, and $\bar{V} = V^* = D(\mu^*)$.

Return now to the \emph{original} trajectory $q(t)$. We know that
$V(q(t))=D(\mu(q(t))) \to D(\mu^*)$; also note from Proposition \ref{prop.unique} that $\mu^*$ is the unique maximizer of $D(\mu)$. Therefore $\mu(t) \to \mu^*$, and consequently $x_{ij}(\mu(t)) \to x^*_{ij}$.  

Returning once again to the version \eqref{eq.first order q} of the queue dynamics, we have a linear first-order system with an input that converges to $\sum_{i}x^*_{ij}+c_j\mu_j^*$. It follows that 
\[
\lim_{t\to \infty} q_j(t) = q_j^* = \sum_{i}x^*_{ij}+c_j\mu_j^*,
\]
the equilibrium queues from \eqref{eq.defqstar}.
\end{proof}

\end{appendices}
	
\bibliographystyle{IEEEtranS}
\bibliography{lbws}

\begin{thebibliography}{10}
\providecommand{\url}[1]{#1}
\csname url@samestyle\endcsname
\providecommand{\newblock}{\relax}
\providecommand{\bibinfo}[2]{#2}
\providecommand{\BIBentrySTDinterwordspacing}{\spaceskip=0pt\relax}
\providecommand{\BIBentryALTinterwordstretchfactor}{4}
\providecommand{\BIBentryALTinterwordspacing}{\spaceskip=\fontdimen2\font plus
\BIBentryALTinterwordstretchfactor\fontdimen3\font minus
  \fontdimen4\font\relax}
\providecommand{\BIBforeignlanguage}[2]{{%
\expandafter\ifx\csname l@#1\endcsname\relax
\typeout{** WARNING: IEEEtranS.bst: No hyphenation pattern has been}%
\typeout{** loaded for the language `#1'. Using the pattern for}%
\typeout{** the default language instead.}%
\else
\language=\csname l@#1\endcsname
\fi
#2}}
\providecommand{\BIBdecl}{\relax}
\BIBdecl

\bibitem{arrow1958studies}
K.~Arrow, L.~Hurwitz, and H.~Uzawa, \emph{Studies in Linear and Non-Linear
  Programming}.\hskip 1em plus 0.5em minus 0.4em\relax Stanford University
  Press, Stanford, California, 1958.

\bibitem{boyd2004convex}
S.~Boyd and L.~Vandenberghe, \emph{Convex optimization}.\hskip 1em plus 0.5em
  minus 0.4em\relax Cambridge university press, 2004.

\bibitem{ephremides1980simple}
A.~Ephremides, P.~Varaiya, and J.~Walrand, ``A simple dynamic routing
  problem,'' \emph{IEEE transactions on Automatic Control}, vol.~25, no.~4, pp.
  690--693, 1980.

\bibitem{goldsztajn2021proximal}
D.~Goldsztajn and F.~Paganini, ``Proximal regularization for the saddle point
  gradient dynamics,'' \emph{IEEE Transactions on Automatic Control}, vol.~66,
  no.~9, pp. 4385--4392, 2021.

\bibitem{goldsztajn2019proximal}
D.~Goldsztajn, F.~Paganini, and A.~Ferragut, ``Proximal optimization for
  resource allocation in distributed computing systems with data locality,'' in
  \emph{2019 57th Annual Allerton Conference on Communication, Control, and
  Computing (Allerton)}.\hskip 1em plus 0.5em minus 0.4em\relax IEEE, 2019, pp.
  773--780.

\bibitem{low2002internet}
S.~H. Low, F.~Paganini, and J.~C. Doyle, ``Internet congestion control,''
  \emph{IEEE control systems magazine}, vol.~22, no.~1, pp. 28--43, 2002.

\bibitem{PagFAllerton23}
F.~Paganini and A.~Ferragut, ``Dynamic load balancing of selfish drivers
  between spatially distributed electrical vehicle charging stations,'' in
  \emph{59th Allerton Conference}, 2023.

\bibitem{paganini2025tcns}
------, ``Dynamics and optimization in spatially distributed electrical vehicle
  charging,'' \emph{IEEE Transactions on Control of Network Systems}, vol.~12,
  no.~1, pp. 403--415, 2025.

\bibitem{paganini2019optimization}
F.~Paganini, D.~Goldsztajn, and A.~Ferragut, ``An optimization approach to load
  balancing, scheduling and right sizing of cloud computing systems with data
  locality,'' in \emph{2019 IEEE 58th Conference on Decision and Control
  (CDC)}.\hskip 1em plus 0.5em minus 0.4em\relax IEEE, 2019, pp. 1114--1119.

\bibitem{roughgarden2005selfish}
T.~Roughgarden, \emph{Selfish routing and the price of anarchy}.\hskip 1em plus
  0.5em minus 0.4em\relax MIT press, 2005.

\bibitem{sheffi1985urban}
Y.~Sheffi, \emph{Urban transportation networks}.\hskip 1em plus 0.5em minus
  0.4em\relax Prentice-Hall, Englewood Cliffs, NJ, 1985.

\bibitem{srikant2013communication}
R.~Srikant and L.~Ying, \emph{Communication networks: an optimization, control,
  and stochastic networks perspective}.\hskip 1em plus 0.5em minus 0.4em\relax
  Cambridge University Press, 2013.

\bibitem{tantawi1985optimal}
A.~N. Tantawi and D.~Towsley, ``Optimal static load balancing in distributed
  computer systems,'' \emph{Journal of the ACM (JACM)}, vol.~32, no.~2, pp.
  445--465, 1985.

\bibitem{der2022scalable}
M.~Van~der Boor, S.~C. Borst, J.~S. Van~Leeuwaarden, and D.~Mukherjee,
  ``Scalable load balancing in networked systems: A survey of recent
  advances,'' \emph{SIAM Review}, vol.~64, no.~3, pp. 554--622, 2022.

\bibitem{wang2016maptask}
W.~Wang, K.~Zhu, L.~Ying, J.~Tan, and L.~Zhang, ``Maptask scheduling in
  mapreduce with data locality: Throughput and heavy-traffic optimality,''
  \emph{IEEE/ACM Trans. Networking}, vol. 24(1), pp. 190--203, 2016.

\end{thebibliography}
	
\end{document}